\documentclass[runningheads,review]{llncs}

\usepackage[none]{hyphenat}
\usepackage{setspace}
\usepackage{latexsym,amsmath}
\usepackage{amsfonts}
\usepackage{algorithm}
\usepackage[noend]{algpseudocode}
\usepackage{graphicx}
\usepackage[dvipsnames]{xcolor}
\usepackage{soul}
\usepackage{comment}
\usepackage{multibib}
\usepackage{url}
\usepackage{bbding}

\usepackage{mathtools}
\DeclarePairedDelimiter\ceil{\lceil}{\rceil}

\newcites{supp}{Appendix References}

\newcounter{claimcounter}

\newcites{supp}{Appendix References}

\makeatletter
\def\BState{\State\hskip-\ALG@thistlm}
\makeatother

\algdef{SE}[SUBALG]{Indent}{EndIndent}{}{\algorithmicend\ }
\algtext*{Indent}
\algtext*{EndIndent}

\includecomment{appendixonly}
\excludecomment{mainonly}

\newcounter{partcounter}
\newcommand{\myprefix}{}

\newcommand{\mylabel}[1]{\label{\myprefix#1}}
\newcommand{\myref}[1]{\ref{\myprefix#1}}
\newcommand{\myeqref}[1]{\eqref{\myprefix#1}}

\newcommand{\setprefix}[1]{\renewcommand{\myprefix}{#1}}

\makeatletter
\newcommand{\mypartprefix}{%
  \ifnum\value{partcounter}>0 
    \Alph{partcounter}%
  \fi
}

\renewcommand{\thefigure}{\mypartprefix\arabic{figure}}

\newcommand{\mypart}{
  \stepcounter{partcounter}  
  \setcounter{section}{0}
  \setcounter{theorem}{0}
  \setcounter{lemma}{0}
  \setcounter{figure}{0}
  \setcounter{algorithm}{0}
  \setcounter{definition}{0}
}

\numberwithin{equation}{section}

\begin{document}
\title{Multimodal Search on a Line\thanks{This is the full version of the paper which appeared in the Proceedings of SIROCCO 2025: The 32nd International Colloquium On Structural Information and Communications Complexity, 2-4 June
2025 Delphi, Greece.}}

\author{
Jared Coleman\inst{1}\textsuperscript(\Envelope\textsuperscript) 
\and
Dmitry Ivanov\inst{2}
\and
Evangelos Kranakis\inst{2}~\inst{5}
\and
Danny Krizanc\inst{3}
\and
Oscar Morales Ponce\inst{4}
}

\institute{
Loyola Marymount University, California, USA \\
\email{jared.coleman@lmu.edu}
\and
Carleton University, Ottawa, Ontario, Canada \\
\email{dimaivanov@cmail.carleton.ca} \\
\email{kranakis@scs.carleton.ca}
\and
Wesleyan University, Middletown CT, USA \\
\email{dkrizanc@wesleyan.edu}
\and
California State University, Long Beach, USA \\
\email{Oscar.MoralesPonce@csulb.edu}
\and
Research supported in part by NSERC Discovery grant.
}

\authorrunning{J. Coleman et al.}

\maketitle

\begin{abstract} 
    Inspired by the diverse set of technologies used in underground object detection and imaging, we introduce a novel multimodal linear search problem whereby a single searcher starts at the origin and must find a target that can only be detected when the searcher moves through its location using the correct of $p$ possible search modes.
    The target's location, its distance $d$ from the origin, and the correct search mode are all initially unknown to the searcher.
    
    We prove tight upper and lower bounds on the competitive ratio for this problem.
    Specifically, we show that when $p$ is odd, the optimal competitive ratio is given by $2p+3+\sqrt{8(p+1)}$, whereas when $p$ is even, the optimal competitive ratio is given by $c$: the unique solution to $(c-1)^4-4p(c+1)^2(c-p-1)=0$ in the interval $\left[2p+1+\sqrt{8p},\infty\right)$. This solution $c$ has the explicit bounds $2p+3+\sqrt{8(p-1)}\leq c\leq 2p+3+\sqrt{8p}$.
    The optimal algorithms we propose require the searcher to move infinitesimal distances and change directions infinitely many times within finite intervals.
    To better suit practical applications, we also propose an approximation algorithm with a competitive ratio of $c+\varepsilon$ (where $c$ is the optimal competitive ratio and $\varepsilon > 0$ is an arbitrarily small constant).
    This algorithm involves the searcher moving finite distances and changing directions a finite number of times within any finite interval. 
\vspace{0.25cm}


\end{abstract}


\mypart
\setprefix{}
\excludecomment{mainonly}
\includecomment{appendixonly}

\section{Introduction}
\mylabel{sec:intro}

Underground object detection and imaging is used in a wide variety of applications in archaeology, urban planning, environmental science, and more.
Many different technologies, each with their own unique capabilities, advantages, and disadvantages, are used to detect different types of materials in different environments.
Ground penetrating radars (GPR), for example, use radar pulses to image the subsurface and are used to locate pipes, cables, and other structures~\cite{uod:gpr}.
Electromagnetic induction (EMI) involves measuring electromagnetic fields to detect conductive materials (i.e., metal objects)~\cite{uod:emi}.
Magnetometers measure variations in the earth's magnetic field caused by ferrous objects and Electrical Resistivity Tomography (ERT) involves measuring electrical resistivity to detect underground soil and groundwater characteristics~\cite{uod:ert}.

Inspired by this diverse set of technologies, we study a novel variant of the classical linear search problem that features a new parameter $p$: the number of search modes available to the searcher, of which the target is only detectable in one.

\subsection{Model, Notation, and Terminology}

A point-like searcher starts at the origin of the real number line and can move with maximum speed $1$.
A stationary target is located an unknown distance $d \geq 1$ from the origin in an unknown direction.
The searcher has $p$ search modes and can only detect the target if the two become collocated while the searcher is in the correct search mode, which is not known in advance. 
The searcher can switch between modes and switch directions instantaneously, though it requires at least $L$ time to search an interval of length $L$ for every mode used.
Note that for $p=1$, this problem is equivalent to the standard linear search problem described in~\cite{beck1970linearsearch}.

For a given algorithm $\mathcal{A}$, let $P_{\mathcal{A}}(t)$ denote the position of the searcher executing algorithm $\mathcal{A}$ at time $t \geq 0$.
For simplicity and without loss of generality, we choose the positive direction to be that for which the searcher first travels a distance of $1$ from the origin. Let $E_{\mathcal{A}}(x,t)$ be the event that the signed coordinate $x$ has been explored in all $p$ modes by the searcher using algorithm $\mathcal{A}$ at or before time $t$. We define $E_{\mathcal{A}}(x,t)$ to always be true if $|x|<1$. 
Henceforth, we use the word ``explored'' to mean ``explored in all $p$ modes.''

Our goal is to minimize the \textit{competitive ratio} of our algorithms --- essentially how many times longer the searcher takes to find the target compared to an omniscient searcher that knows the target's location and the correct search mode ahead of time. We define:
\begin{equation*}
    T_{\mathcal{A}}(x) = \inf\left\{t\geq0 : E_{\mathcal{A}}(x,t)\right\},\quad CR_{\mathcal{A}}(x) = \frac{1}{|x|}T_{\mathcal{A}}(x),\quad CR_{\mathcal{A}} = \sup_{x} CR_{\mathcal{A}}(x) .
\end{equation*}
We consider only ``successful'' algorithms: those that yield finite exploration times for all coordinates $x$.

\subsection{Related work}

Stochastic linear search was independently considered by Bellman and Beck in the 1960's (cf. \cite{beck1964linear,bellman1963optimal}). Work in linear search dates back to the 1970s with Beck and Newman who introduced the linear search problem and proposed an algorithm with optimal competitive ratio $9$~\cite{beck1970linearsearch}.
Since then, many variants of the linear search problem have been studied and the competitive ratios of various algorithms have been analyzed~\cite{gal_search_games,baeza_yates,group_search,search_plane}. 

The multimodal linear search problem we study is a generalization of linear search (when the number of modes $p=1$, the two problems are equivalent).
Despite the vast literature in search theory, work on multimodal search is sparse.
A predator-prey model was introduced in~\cite{knoppien1985predators} where the predator has two modes of searching: a fast mode where the predator has a low probability of detecting the prey when encountered and a slow mode where the predator has a high probability of detecting the prey.
Similarly in~\cite{benichou2011intermittent}, the authors study problems where the searcher moves at a faster speed at which the searcher cannot detect the target and a slower speed at which the searcher can detect the target (with certainty).
The Beachcomber problem~\cite{beachcomber} extends this idea to a multi-agent setting where searchers must cooperate to find the target in as little time as possible.
Another key similarity between these studies and ours is that, depending on the strategy the searcher uses, it may not detect the target when passing through its location.
This is reminiscent of work in search with uncertain detection where the searcher may miss the target with some probability~\cite{gal_search_games,probabilistic_faults}.  

To the best of our knowledge, this work is the first to consider the multimodal linear search problem.

\subsection{Outline and results of the paper}

The rest of the paper is organized as follows.
In Section~\ref{sec:upper_bounds}, we propose an algorithm and prove an upper bound on its competitive ratio of $2p+3+\sqrt{8(p+1)}$ when $p$ is odd and $c$ (where $c$ is the unique solution to $(c-1)^4-4p(c+1)^2(c-p-1)=0$ in the interval $\left[2p+1+\sqrt{8p},\infty\right)$) when $p$ is even.
This solution $c$ has the explicit bounds $2p+3+\sqrt{8(p-1)}\leq c\leq 2p+3+\sqrt{8p}$.
Then, in Section~\ref{sec:lower_bounds}, we prove matching lower bounds on the competitive ratio for any algorithm, thus proving our upper bounds are tight.
The algorithms we propose in Section~\ref{sec:upper_bounds} involve the searcher moving  infinitesimal distances and changing directions infinitely many times within finite intervals. 
To better suit practical applications, we introduce in Section~\ref{sec:approx}  modified algorithms (for both the odd and even cases) with a competitive ratio of $c + \varepsilon$ (where $c$ is the optimal competitive ratio and $\varepsilon > 0$ is an arbitrarily small constant).
This modified algorithm involves the searcher moving finite distances and changing directions a finite number of times within any finite interval.
Numerical representations for the bounds proved in Section~\ref{sec:upper_bounds} 
are displayed in Table~\myref{tbl:crs}.
\begin{table}[ht]
\caption{Competitive ratio of multimodal search with $p$ modes, for $p=1,2,\ldots, 16$.}
\centering
\begin{tabular}{|c|c|c|c|}
    \hline
    \multicolumn{2}{|c|}{Odd $p$} & \multicolumn{2}{c|}{Even $p$} \\
    \hline
    $p$ & Competitive Ratio & $p$ & Competitive Ratio \\
    \hline
    1   &   9           &   2   &   10.27303    \\
    3   &   14.65685    &   4   &   16.08120    \\
    5   &   19.92820    &   6   &   21.43387    \\
    7   &   25          &   8   &   26.55911    \\
    9   &   29.94427    &   10  &   31.54214    \\
    11  &   34.79796    &   12  &   36.42569    \\
    13  &   39.58301    &   14  &   41.23468    \\
    15  &   44.31371    &   16  &   45.98516    \\
    \hline
\end{tabular}
\mylabel{tbl:crs}
\end{table}
Finally, in Section~\ref{sec:conclusion}, we summarize our contributions and discuss potential avenues for future research.
\begin{mainonly}
All proofs omitted due to space constraints can be found in the full version of the paper~\cite{full_paper}.
\end{mainonly}

\section{Upper Bounds}
\mylabel{sec:upper_bounds}

In this section, we propose algorithms  (one for odd $p$  and one for even $p$) for solving the multimodal linear search problem and prove upper bounds on their competitive ratio.

\subsection{Preliminaries}

The algorithms that we present follow a certain principle: select an interval and explore it with all $p$ modes by moving back and forth before proceeding.
We refer to these intervals as ``cells.'' The parity of $p$ matters because the searcher will finish searching a cell on the opposite side of the cell that it started for odd $p$ whereas it will finish at the same position that it started for even $p$. 
This can be beneficial or detrimental depending on where the searcher plans on going next. 
For odd $p$, ending on the opposite side facilitates searching multiple cells in a row but also means that the searcher will have to backtrack unproductively when returning to the origin.
Meanwhile, for even $p$, the searcher finishes exploring the last cell while also making its way back to the origin; however, any prior cells must be traversed an extra time to reach the start of the next cell. 
Figure~\ref{fig:thorough} depicts these advantages/disadvantages for the odd/even cases.
This reflects a fundamental disparity between the two cases that is  evident in our results. 

We formalize this in the \textsc{CellSearch} procedure, Algorithm \ref{alg:cellsearch}, below. Note that we always begin cell search at the point closest to the origin. 

\begin{algorithm}[H]
  \caption{\mylabel{alg:cellsearch}(Cell Search)}
  \begin{algorithmic}[1]
    \Procedure{CellSearch}{$x_{init}$, $\delta$, $p$}
      \Statex $x_{init}, \delta, p$: current position, positive distance, number of search modes 
      \State \textbf{for} $i \gets 0, ..., p-1$ \textbf{do} \Indent
      \State travel to current position plus $(-1)^i\delta\cdot \text{sign}(x_{init})$ while searching in mode $i$\EndIndent
      \State \textbf{endfor}
    \EndProcedure
  \end{algorithmic}
\end{algorithm}

Observe that splitting intervals into smaller cells enables points near the origin to be explored as soon as possible.
Imagine a procedure that searches an interval by dividing it into a number of cells with a fixed size (except for the last cell).
We call this procedure \textsc{DiscreteThoroughSearch}.

\begin{algorithm}[H]
  \caption{\mylabel{alg:discrete-thorough}(Discrete Thorough Search)}
  \begin{algorithmic}[1]
    \Procedure{DiscreteThoroughSearch}{$x_{init}$, $\delta$, $s$, $p$}
      \Statex $x_{init}$: current position
      \Statex $\delta$: positive distance
      \Statex $s$: cell size
      \Statex $p$: number of search modes
      \State $x \gets 0$
      \State \textbf{while} $x < \delta$ \textbf{do} \Indent
      \State $\Delta x \gets \min(s, \delta - x)$
      \State \textbf{call} \Call{CellSearch}{$x_{init} + x \cdot \text{sign}(x_{init})$, $\Delta x$, $p$}
      \State \textbf{if} $p$ is even \textbf{then} travel to $x + \Delta x$ \textbf{endif}
      \State $x \gets x + \Delta x$\EndIndent
      \State \textbf{endwhile}
    \EndProcedure
  \end{algorithmic}
\end{algorithm}

There is no theoretical lower limit on the cell size $s$.
In fact, as $s$ approaches $0$, \textsc{DiscreteThoroughSearch} essentially involves the agent exploring the interval in all search modes simultaneously at a slower speed $\frac{1}{p}$ if $p$ is odd and $\frac{1}{p+1}$ if $p$ is even.
We permit this behavior and simply call it \textsc{ThoroughSearch}.

\begin{algorithm}[H]
\caption{\mylabel{alg:thorough}(Thorough Search)}
    \begin{algorithmic}[1]
        \Procedure{ThoroughSearch}{$x_{init}$, $\delta$, $p$}
            \Statex $x_{init}$: current position
            \Statex $\delta$: positive distance
            \Statex $p$: number of search modes
            \State $u \gets$ $\frac{1}{p}$ \textbf{if} $p$ \textbf{is odd} \textbf{else} $\frac{1}{p+1}$
            \State travel to $x_{init} + \delta \cdot \text{sign}(x_{init})$ at speed $u$ while searching all modes simultaneously
        \EndProcedure
    \end{algorithmic}
\end{algorithm}

\begin{figure}[!htb]
    \begin{center}
      \includegraphics[width=0.3\textwidth]{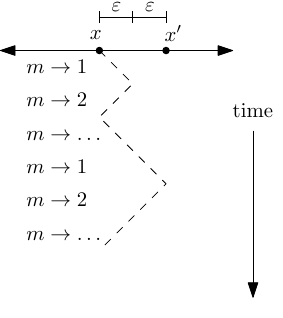}
      \hspace{1cm}
      \includegraphics[width=0.3\textwidth]{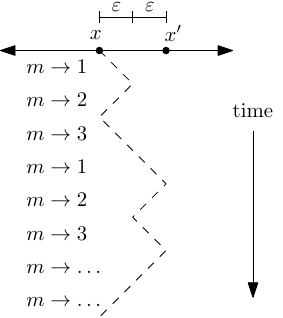}
    \end{center}
    \caption{The searcher explores the segment $[x,x^\prime]$ using \textsc{DiscreteThoroughSearch} with two $\varepsilon-$width cells. The left figure depicts the case where $p=2$ (even) and the right figure depicts the case where $p=3$ (odd). The value of $m$ represents the mode the searcher is operating in and $m \rightarrow \ldots$ is used to indicate that the searcher is traversing a cell that has already been fully explored.}
    \mylabel{fig:thorough}
\end{figure}




\subsection{Upper Bound for Odd p}

For the case of odd $p$, we apply \textsc{ThoroughSearch} on intervals of exponentially growing size, alternating between the positive and negative directions. The result is Algorithm \ref{alg:odd}.
\begin{algorithm}[H]
  \caption{\mylabel{alg:odd}(Multimodal search for odd $p$)}
  \begin{algorithmic}[1]
      \State \textbf{input:} number of modes $p$ and growth factor $a>1$.
      \State terminate upon detecting the target
      \State \textbf{for} $i \gets 0,1,2, ..., \infty$ \textbf{do} \Comment{each iteration $i$ represents a round} \mylabel{alg:odd:line:iter} \Indent
      \State $x \gets (-1)^ia^{i-1}$
      \State go to $x$
      \State \textbf{call} \Call{ThoroughSearch}{$x$, $\left(a^{i+1} - a^{i-1}\right)$, $p$}
      \State return to the origin \EndIndent
      \State \textbf{endfor}
  \end{algorithmic}
\end{algorithm}

\begin{theorem} \mylabel{thm:odd_upbound}
  Algorithm \myref{alg:odd} satisfies $CR_{\mathcal{A}} = 2p+3+\sqrt{8(p+1)}$ with the growth factor $a=1+\sqrt{\frac{2}{p+1}}$.
\end{theorem}
\begin{proof}
  Algorithm \myref{alg:odd} plays out in rounds numbered $i\geq 0$ (line~\myref{alg:odd:line:iter}). 
  In each round $i$, the searcher explores the interval from $a^{i-1}$ to $a^{i+1}$ in direction $(-1)^i$ using \textsc{ThoroughSearch}. Note that for round $0$ only, some coordinates $|x|<1$ are searched superfluously, though this does not significantly affect the upcoming analysis. Executing \textsc{ThoroughSearch} takes $p(a^{i+1}-a^{i-1}) = pa^{i-1}(a^2-1)$ time and an additional $(a^{i+1}+a^{i-1}) = a^{i-1}(a^2+1)$ time is spent traveling back and forth from the origin. In total, completing a round $i$ takes $t_i = a^{i-1}\left((p+1)(a^2-1)+2\right)$ time. Completing all rounds from $0$ to $i$ takes $T_i = \frac{a^{i+1}-1}{a(a-1)}\left((p+1)(a^2-1)+2\right)$ time. Coordinates that are explored on round $i$ can be expressed in the form $(-1)^ia^{i-1}\left(1+x\right)$ where $0<x\leq a^2-1$. For rounds $0$ and $1$ only, $x$ may be exactly $0$; this does not affect the upcoming analysis. Each coordinate is explored at time $T_{i-1}+a^{i-1}\left(1+px\right)$.
  \begin{equation*}
    CR_{\mathcal{A}}(x) = \frac{T_{i-1}+a^{i-1}\left(1+px\right)}{a^{i-1}(1+x)} = p + \frac{T_{i-1}-(p-1)a^{i-1}}{a^{i-1}(1+x)}
  \end{equation*}
  If $T_{i-1} \leq (p-1)a^{i-1}$, then the ratio is maximized by selecting $x=a^2-1$ and is no greater than $p$. Otherwise, the ratio is maximized by selecting $x$ arbitrarily close to $0$ (or simply equal to $0$ in the case of the first two rounds). We now check the latter case for all rounds. What we will find is that the worst-case ratio is produced by arbitrarily late rounds and that it exceeds $p$, meaning that the former case does not produce the worst-case ratio.
  \begin{align*}
    \lim_{x\to 0^+}CR_{\mathcal{A}}(x) &= 1+\frac{T_{i-1}}{a^{i-1}} = 1+\frac{1-a^{-i}}{a-1}\left((p+1)(a^2-1)+2\right) \\
    \lim_{i\to\infty}\;\lim_{x\to 0^+}CR_{\mathcal{A}}(x) &= 1+\frac{1}{a-1}\left((p+1)(a^2-1)+2\right) = CR_{\mathcal{A}}
  \end{align*}
  Selecting the growth factor $a = 1 + \sqrt{\frac{2}{p+1}}$ yields
  \begin{equation*}
    CR_{\mathcal{A}} = 1+(p+1)\left(2+\sqrt{\frac{2}{p+1}}\right)+2\left(\sqrt{\frac{2}{p+1}}\right)^{-1} = 2p+3+\sqrt{8(p+1)}
  \end{equation*}
  Indeed, this ratio is greater than $p$. This proves Theorem \myref{thm:odd_upbound}.
\end{proof}

\subsection{Upper Bound for Even p}
In this section we give an algorithm with optimal competitive ratio for the case where $p$ is even. 
The optimal ratio is given implicitly as a root of the polynomial $D_p(c) = (c-1)^4 - 4p(c+1)^2(c-p-1)$. 
By Theorem \myref{thm:odd_lowbound} (presented in Section~\ref{sec:lower_bounds} with other lower bounds), the competitive ratio for the even case must be at least $2p+1+\sqrt{8p}$ since the same problem with $p-1$ modes instead of $p$ is strictly easier for the searcher.
For this reason, we need only consider the function $D_p(c)$ for $c\geq 2p+1+\sqrt{8p}$. We derive the behavior of $D_p(c)$ on this interval in Lemma \myref{lem:discrim_behavior}.
\begin{lemma} \mylabel{lem:discrim_behavior}
  For even $p$, the function $D_p(c) = (c-1)^4 - 4p(c+1)^2(c-p-1)$ transitions from negative to positive exactly once within the range $\left[2p+1+\sqrt{8p},\infty\right)$. 
\end{lemma}
\begin{proof}
  Consider the polynomial $D_p(c)$ as well as its second and third derivatives.
  \begin{align*}
    D_p(c) &= (c-1)^4 - 4p(c+1)^2(c-p-1) \\
    &= c^4 - 4(p+1)c^3 + (4p(p-1)+6)c^2 + 4(2p-1)(p+1)c + (2p+1)^2 \\
    D_p^{(2)}(c) &= 12c^2 - 24(p+1)c + 2(4p(p-1)+6) \\
    &= 12c(c-2(p+1)) + 2(4p(p-1)+6) \\
    D_p^{(3)}(c) &= 24c - 24(p+1)
  \end{align*}
  We note that $D_p^{(3)}(c) > 0$ for all $c \geq 2p+1+\sqrt{8p} > p+1$. We also note that since $p\geq 2$, $D_p^{(2)}(c) > 0$ for all $c \geq 2p+1+\sqrt{8p} > 2(p+1)$. This shows that $D_p(c)$ is concave up for $c \geq 2p+1+\sqrt{8p}$. We now show that $D_p(2p+1+\sqrt{8p})<0$.
  \begin{align*}
    D_p(2p+1+\sqrt{8p}) &= (2p+\sqrt{8p})^4 - 4p(2p+2+\sqrt{8p})^2(p+\sqrt{8p}) \\
    &= 16p^{\frac{3}{2}}\left(\sqrt{p}(\sqrt{p}+\sqrt{2})^4 - (p+1+\sqrt{2p})^2(\sqrt{p}+\sqrt{8})\right) \\
    &= 16p^{\frac{3}{2}}\left(-p\sqrt{8}-5\sqrt{p}-\sqrt{8}\right) < 0
  \end{align*}
  Since $D_p(2p+1+\sqrt{8p})<0$ and $D_p(c)$ has the dominant term $c^4$, it must eventually become positive. Since it is concave up beyond $c=2p+1+\sqrt{8p}$, it cannot switch signs more than once. This proves Lemma \myref{lem:discrim_behavior}.
\end{proof}

Again, our algorithm alternates searching intervals of growing size in each direction but in this case, a fraction $r$ of the interval is searched using \textsc{ThoroughSearch} and the remainder is treated as single cell. The result is Algorithm \ref{alg:even}.
\begin{algorithm}[H]
  \caption{\mylabel{alg:even}(Multimodal search for even $p$)}
  \begin{algorithmic}[1]
      \State \textbf{input:} number of modes $p$, growth factor $a>1$, and a coefficient $0\leq r\leq1$.
      \State terminate at any point if collocated with target
      \State \textbf{for} $i \gets 0,1,2, ..., \infty$ \textbf{do} \Comment{each iteration $i$ represents a round} \mylabel{alg:even_add:line:iter} \Indent
      \State $x_{start} \gets (-1)^ia^{i-1}$
      \State $x_{end} \gets (-1)^ia^{i+1}$
      \State $x_{mid} \gets (1-r)x_{start} + rx_{end}$
      \State go to $x_{start}$
      \State \textbf{call} \Call{ThoroughSearch}{$x_{start}$, $\left|x_{mid}-x_{start}\right|$, $p$}
      \State \textbf{call} \Call{CellSearch}{$x_{mid}$, $\left|x_{end}-x_{mid}\right|$, $p$}
      \State return to the origin \EndIndent
      \State \textbf{endfor}
  \end{algorithmic}
\end{algorithm}

\begin{theorem} \mylabel{thm:even_upbound}
  Algorithm \myref{alg:even} achieves a competitive ratio $CR_{\mathcal{A}}=c^*$ where $c^*$ is the unique solution to $D_p(c)=0$ in the range $\left[2p+1+\sqrt{8p},\infty\right)$. The parameters $a$ and $r$ needed to do so are given in terms of the above value $c^*$ as follows:
  \begin{align*}
    a &= \frac{(c^*-1)^2}{2p(c^*+1)}, ~~~
    r = \frac{\sqrt{\left(p(a^2-1)+2\right)^2+8p(a^2-1)(a-1)}-2}{4(a^2-1)}-\frac{p}{4}.
  \end{align*}
\end{theorem}
\begin{proof}
  Algorithm \myref{alg:even} plays out in rounds numbered $i\geq 0$ (line~\myref{alg:even_add:line:iter}). Consider an arbitrary round $i$ where the searcher goes in direction $(-1)^i$. The searcher begins by going to $a^{i-1}$ which is at the boundary of everything searched thus far. The searcher then executes \textsc{ThoroughSearch} to search in front of itself up to the point $m_i = a^{i-1} + r\left(a^{i+1}-a^{i-1}\right) = a^{i-1}\left(r(a^2-1)+1\right)$. This will be dubbed phase 1. The remaining interval from $m_i$ to $a^{i+1}$ is searched by moving back and forth using a single call to \textsc{CellSearch}. This will be dubbed phase 2. The searcher ends phase 2 at position $m_i$ and then returns to the origin to end the round. Note that for round $0$ only, some coordinates $|x|<1$ are searched superfluously, though this does not significantly affect the upcoming analysis.

  We define $t_i$ to be the time taken to complete round $i$ and $T_i = \sum_{j=0}^{i}t_j$. The above description implies that
  \begin{align*}
    t_i &= a^{i-1} + (p+1)(m_i-a^{i-1}) + p(a^{i+1}-m_i) + m_i\\
    &= 2a^{i-1} + p\left(a^{i+1}-a^{i-1}\right) + 2r\left(a^{i+1}-a^{i-1}\right)\\
    &= a^{i-1}\left(2 + (p+2r)\left(a^2-1\right)\right)~~ \text{and} \\
    T_i &= \frac{a^{i+1}-1}{a(a-1)}\left(2 + (p+2r)\left(a^2-1\right)\right) .
  \end{align*}
  We now examine the competitive ratio of coordinates searched during phase 1. These are coordinates of the form $d_1(i,x) = (-1)^ia^{i-1}(1+x)$ for $x\in\left(0,a^2-1\right]$. For rounds $0$ and $1$ only, $x$ may be exactly $0$; this does not affect the upcoming analysis.
  \begin{align*}
    T_{\mathcal{A}}(d_1(i,x)) &= T_{i-1} + a^{i-1} + a^{i-1}(p+1)x \\
    CR_{\mathcal{A}}(d_1(i,x)) &= 1 + \frac{T_{i-1} + a^{i-1}px}{a^{i-1}(1+x)} \\
    &= 1 + \frac{\frac{a^i-1}{a(a-1)}\left(2 + (p+2r)\left(a^2-1\right)\right) + a^{i-1}px}{a^{i-1}(1+x)} \\
    &= 1 + \frac{\frac{1-a^{-i}}{a-1}\left(2 + (p+2r)\left(a^2-1\right)\right) + px}{1+x} \\
    \lim_{i\to \infty}CR_{\mathcal{A}}(d_1(i,x)) &= 1 + \frac{\frac{2}{a-1} + (p+2r)(a+1) + px}{1+x}  \\
    &= 1+p + \frac{\frac{2}{a-1} + (p+2r)a + 2r}{1+x}
  \end{align*}
  Since the numerator is positive, this value is maximized by minimizing $x$. We select $x$ arbitrarily close to $0$ (or simply equal to $0$ in the case of the first two rounds). 
  \begin{equation}
    \lim_{i\to \infty}\lim_{x\to 0^+}CR_{\mathcal{A}}(d_1(i,x)) = 1 + \frac{2}{a-1} + (p+2r)(a+1) . \mylabel{eq:upbound_cr1}
  \end{equation}
  Note that this quantity is increasing with respect to $r$.
  Next, we consider coordinates searched during phase 2: $d_2(i,x) = (-1)^ix$ for $x\in\left(m_i,a^{i+1}\right]$.
  \begin{align*}
    T_{\mathcal{A}}(d_2(i,x)) &= T_{i} - x, ~~~
    CR_{\mathcal{A}}(d_2(i,x)) = \frac{T_{i}}{x} - 1 .
  \end{align*}
  This value is maximized as $x$ approaches $m_i$.
  \begin{align*}
    \lim_{x\to m_i^+}CR_{\mathcal{A}}(d_2(i,x)) &= \frac{T_i}{m_i} - 1 
    = \frac{\frac{a^{i+1}-1}{a(a-1)}\left(2 + (p+2r)\left(a^2-1\right)\right)}{a^{i-1}\left(r(a^2-1)+1\right)} - 1 \\
    &= \frac{\frac{a-a^{-i}}{a-1}\left(2 + (p+2r)\left(a^2-1\right)\right)}{r(a^2-1)+1} - 1 \\
    \lim_{i\to \infty}\lim_{x\to m_i^+}CR_{\mathcal{A}}(d_2(i,x)) &= \frac{\frac{a}{a-1}\left(2 + (p+2r)\left(a^2-1\right)\right)}{r(a^2-1)+1} - 1 \\
    &= \frac{\frac{a}{a-1}p\left(a^2-1\right)}{r(a^2-1)+1} + \frac{2a}{a-1} - 1
  \end{align*}
  Note that this quantity decreases with respect to $r$.
  Since the phase 1 ratio increases with respect to $r$ while the phase 2 ratio decreases, the maximum of the two is minimized when they are equal.
  \begin{align*}
    1 + \frac{2}{a-1} + (p+2r)(a+1) &= \frac{\frac{a}{a-1}p\left(a^2-1\right)}{r(a^2-1)+1} + \frac{2a}{a-1} - 1 \\
    (p+2r)(a+1) &= \frac{\frac{a}{a-1}p\left(a^2-1\right)}{r(a^2-1)+1} \\
    (p+2r)\left(r(a^2-1)+1\right) &= ap \\
    2(a^2-1)r^2 + \left(p(a^2-1)+2\right)r + p(1-a) &= 0
  \end{align*}
  \begin{equation*}
    r = \frac{\sqrt{\left(p(a^2-1)+2\right)^2+8p(a^2-1)(a-1)}-2}{4(a^2-1)}-\frac{p}{4}
  \end{equation*}
  Substituting this into Equation \myeqref{eq:upbound_cr1} yields
  \begin{equation*}
    1 + \frac{p(a+1)}{2}+\frac{2+\sqrt{\left(p(a^2-1)+2\right)^2+8p(a^2-1)(a-1)}}{2(a-1)} .
  \end{equation*}
  This function denotes the best competitive ratio achievable by Algorithm \myref{alg:even} for a given growth ratio $a$. Minimizing this function with respect to $a$ yields the optimal parameters for Algorithm \myref{alg:even}. Suppose that we desire some competitive ratio $c$:
  \begin{align*}
    c = 1 + \frac{p(a+1)}{2}+\frac{2+\sqrt{\left(p(a^2-1)+2\right)^2+8p(a^2-1)(a-1)}}{2(a-1)} \\
    2(c-1)(a-1) - p(a^2-1) - 2 = \sqrt{\left(p(a^2-1)+2\right)^2+8p(a^2-1)(a-1)}
  \end{align*}
  For simplicity, let us define $f(a)=p(a^2-1)+2$.
  \begin{align*}
    2(c-1)(a-1) - f(a) &= \sqrt{f(a)^2+8p(a+1)(a-1)^2} \\
    4(c-1)^2(a-1)^2 - 4(c-1)(a-1)f(a) + f(a)^2 &= f(a)^2+8p(a^2-1)(a-1) \\
    (c-1)^2(a-1) - (c-1)\left(p(a^2-1)+2\right) &= 2p(a^2-1) \\
    (c-1)^2(a-1) - p(c-1)(a^2-1) - 2(c-1) &= 2p(a^2-1) \\
    (c-1)^2(a-1) - 2(c-1) &= p(c+1)(a^2-1) \\
    (c-1)^2a - (c-1)^2 - 2(c-1) &= p(c+1)a^2 - p(c+1) \\
    p(c+1)a^2 - (c-1)^2a + (c+1)\left(c-1-p\right) &= 0
  \end{align*}
  Note that this quadratic has discriminant $D_p(c)$. We conclude that Algorithm \myref{alg:even} can attain competitive ratio $c$ so long as $D_p(c) \geq 0$. According to Lemma \myref{lem:discrim_behavior}, the lowest attainable ratio is $c^*$, the unique solution to $D_p(c) = 0$ in the range $\left[2p+1+\sqrt{8p},\infty\right)$. The parameters to achieve this ratio are given by
  \begin{align*}
    a &= \frac{(c^*-1)^2}{2p(c^*+1)}, ~~~
    r = \frac{\sqrt{\left(p(a^2-1)+2\right)^2+8p(a^2-1)(a-1)}-2}{4(a^2-1)}-\frac{p}{4} .
  \end{align*}
  This concludes the proof of Theorem~\myref{thm:even_upbound}.
\end{proof}

\section{Lower Bounds}
\mylabel{sec:lower_bounds}

\subsection{Preliminaries}
When considering all possible algorithms, we permit infinitely precise motion such as that performed by \textsc{ThoroughSearch}, but in a more general form. To determine whether a pattern of motion is valid, we imagine a scenario where motion is restricted to small, discrete steps of size $\varepsilon$ in which the searcher uses one search mode and moves a distance $\varepsilon$ with speed $1$ in either direction. A pattern of motion is valid if it can be approximated arbitrarily well in this restricted case as $\varepsilon$ tends to $0$.

We define several concepts that are integral to proving lower bounds for the competitive ratio of all algorithms. For a given algorithm $\mathcal{A}$, let
\begin{align*}
  I^\sigma_{\mathcal{A}}(t) &= (-1)^\sigma \sup \left\{x : x>0, \forall 0\leq x'\leq x\; E_{\mathcal{A}}\left((-1)^\sigma x',t\right)\right\} .
\end{align*}
For a given $t$, the interval $\left(I^1_{\mathcal{A}}(t), I^0_{\mathcal{A}}(t)\right)$, hereafter referred to as the ``central island,'' is the greatest-length contiguous interval around the origin such that all coordinates within the interval have been explored by algorithm $\mathcal{A}$ by time $t$. 
We are most interested in the length of the island; hence, the inclusion or exclusion of endpoints does not matter. 
Also, note that any integer values of $\sigma$ may be used and only their parity is relevant.
Even values represent the positive side of the line and odd values the negative side.

We define a function that extracts sequences of key values from an arbitrary algorithm. The existence of these key values and their properties will be used to set lower bounds on the performance of all algorithms. The function $X(\mathcal{A})$ takes an algorithm $\mathcal{A}$ and produces a tuple of sequences $\left(\left(x_i\right)_{i=-2}^{\infty}, \left(t_i\right)_{i=0}^{\infty}\right)$ where the former is a sequence of signed coordinates and the latter is a sequence of time values.

To obtain the two sequences, we first define $t_0$ to be the time at which algorithm $\mathcal{A}$ first reaches $x=1$. We also define $x_{-2}=1$, $x_{-1} = -1$. Then, for each $i\geq 0$ in order of increasing $i$, we select $t_{i+1}$ and $x_i$ such that
  $t_{i+1} = \sup \left\{t : I^{i+1}_{\mathcal{A}}(t) = I^{i+1}_{\mathcal{A}}(t_i)\right\} ,~
  x_i = \lim_{t\to t_{i+1}^-} I^i_{\mathcal{A}}(t)$.
We refer to the time interval $(t_i, t_{i+1}]$ as ``period $i$'' (not to be confused with the prior notion of rounds, though there do exist parallels between the two). During period $i$, the central island of explored points is grown in direction $(-1)^i$. Period $i$ ends and period $i+1$ begins precisely when the searcher begins growing the island in the opposite direction. $x_i$ is the signed endpoint of the central island in direction $(-1)^i$ at its furthest extent during period $i$.

The sequence $\left(x_i\right)_{i=-2}^{\infty}$ satisfies a number of key properties. By definition, $|x_{i}| > |x_{i-2}|$ for all $i\geq0$. During each period, the searcher must at least traverse the interval $(-1,1)$ in time $2$ to start the next period. This means that a finite number of periods occur in any finite time. Furthermore, since we only consider ``successful'' algorithms, it must be true that $\lim_{i\to \infty}|x_i| = \infty$ as all finite distances have a finite exploration time.

Finally, we prove Lemma \myref{lem:general_becknewman} which will be used in lower bound proofs.
This lemma shows the impossibility of the existence of infinite sequences whose terms satisfy a certain inequality and is a modification of part of Lemma 2 from \cite{beck1970linearsearch} that gives a more general result.

\begin{lemma} \mylabel{lem:general_becknewman}
For $a>0$, $b>0$, and $a^2-4b<0$, no positive, infinite sequence $\left(y_i\right)_{i=0}^{\infty}$ can satisfy $y_{i+2}\leq ay_{i+1}-by_i$ for all $i$.
\end{lemma}
\begin{appendixonly}
  \begin{proof}
  Let us define another sequence $\left(z_i\right)_{i=0}^{\infty}$ given by $z_i = \left(\frac{a}{2b}\right)^iy_i$. We can establish a bound on the second difference of this sequence.
  \begin{align*}
    & \left(\frac{2b}{a}\right)^{i+2}\left(z_{i+2}-2z_{i+1}+z_i\right) 
    = y_{i+2}-2\left(\frac{2b}{a}\right)y_{i+1}+\left(\frac{2b}{a}\right)^2y_i \\
    &= y_{i+2}-\left(\frac{4b}{a^2}\right)\left(ay_{i+1}-by_i\right) 
    \leq ay_{i+1}-by_i-\left(\frac{4b}{a^2}\right)\left(ay_{i+1}-by_i\right) \\
    &= -\left(\frac{4b-a^2}{a^2}\right)\left(ay_{i+1}-by_i\right) 
    \leq -\left(\frac{4b-a^2}{a^2}\right)y_{i+2} .
  \end{align*}
  Therefore,
  \begin{equation}
    z_{i+2}-2z_{i+1}+z_i \leq \left(\frac{a^2-4b}{a^2}\right)z_{i+2} . \mylabel{eq:general_becknewman_max_concav}
  \end{equation}
  We see that the second difference of the sequence $\left(z_i\right)_{i=0}^{\infty}$ at index $i+2$ is at most the value of the sequence at index $i+2$ multiplied by a negative quantity $\frac{a^2-4b}{a^2}$. This applies for all $i\geq 0$.

  A sequence that is both positive and concave down can never decrease. If it ever did, then the slope could never increase again and the sequence would eventually cross the $y$ axis, contradicting its positivity. This means that the second difference of $\left(z_i\right)_{i=0}^{\infty}$ must tend to $0$, lest the slope decrease to arbitrarily low values. According to Equation \myeqref{eq:general_becknewman_max_concav}, this can only happen if $z_i$ itself tends to $0$. However, being a positive sequence, it must decrease at some point in order to tend to $0$. Hence, the existence of a sequence $\left(z_i\right)_{i=0}^{\infty}$ satisfying the above conditions is contradictory, proving Lemma \myref{lem:general_becknewman}.
  \end{proof}
\end{appendixonly}

\subsection{Lower Bound for Odd p}
\begin{theorem} \mylabel{thm:odd_lowbound}
    For odd $p$, no algorithm $\mathcal{A}$ satisfies $CR_{\mathcal{A}} < 2p+3+\sqrt{8(p+1)}$.
\end{theorem}
\begin{proof}
  Consider an arbitrary algorithm $\mathcal{A}$. We obtain the sequences $\left(x_i\right)_{i=-2}^{\infty}$ and $\left(t_i\right)_{i=0}^{\infty}$ from $X(\mathcal{A})$ as described previously.  
  
  We start by establishing a lower bound on $t_{i+2}$: the time that it takes to complete all of the periods up to and including $i+1$. By the end of period $i+1$, the central island reaches length $|x_{i+1}|+|x_i|$. At $t=0$, the island had length $2$, meaning that at least $p\left(|x_{i+1}|+|x_i|-2\right)$ time has been spent expanding the island across all of the periods before and including $i+1$. One may call this the ``total productive time'' as any further time is spent unproductively retreading old ground. We note that every period $j$ includes a minimum amount of unproductive time. At the start of period $j$, the central island has length $|x_{j-1}|+|x_{j-2}|$, which must be retreaded during period $j$ to start the next period. We call this the ``base retreading time'' of period $j$. The nature of odd $p$ guarantees further retreading. The searcher cannot explore an interval using all $p$ modes and return to the origin without retreading the interval in question at least once because that implies crossing something an odd number of times and returning to the starting side. The interval between $x_{j-2}$ and $x_j$ becomes fully explored by the end of period $j$, meaning that its length must have been retreaded at least once prior to the end of period $j$. The time taken to do so is not part of the base retreading time of period $j$ or any prior period as the latter involves retreading ground on the origin-facing side of $x_{j-2}$ while the former does not. We label this the ``odd retreading time of period $j$.'' In summary, each period $0\leq j\leq i+1$ contributes at least $|x_{j}|+|x_{j-1}|$ of unproductive time. In combination with the total productive time, this yields a lower bound on the time needed to complete all periods up to and including $i+1$. Let us use the shorthand $\hat{x}_i = \sum_{j=0}^{i}|x_j|$.
  \begin{align*}
    t_{i+2} &\geq p\left(|x_{i+1}|+|x_i|\right) + \sum_{j=0}^{i+1}\left(|x_j|+|x_{j-1}|\right) \\
    &= (p+1)\left(|x_{i+1}|+|x_i|\right) + |x_i| + 2\hat{x}_{i-1} + |x_{-1}| \\
    &\geq (p+1)\left(|x_{i+1}|+|x_i|\right) + |x_i| + 2\hat{x}_{i-1} .
  \end{align*}
  We will use this result to evaluate the performance of $\mathcal{A}$. Consider a period $i\geq 0$. Coordinates of the form $x_i + (-1)^i\varepsilon$ remain unexplored for certain, arbitrarily small $\varepsilon>0$. Were this not true, such points would all be explored for sufficiently small values of $\varepsilon$, suggesting that the contiguous island extends further than $x_i$; however, this contradicts the definition of $x_i$. Such points are explored no sooner than period $i+2$. Generously, we assume that all coordinates explored during period $i+2$ are explored at the start of the period: $t_{i+2}$. 
  \begin{align*}
    CR_{\mathcal{A}}\left(x_i + (-1)^i\varepsilon\right) &\geq \frac{t_{i+2}}{|x_i|+\varepsilon} .
  \end{align*}
  Since we can select arbitrarily small values of $\varepsilon$, it is in fact necessary that
  \begin{align*}
    CR_{\mathcal{A}}\left(x_i + (-1)^i\varepsilon\right) &\geq \frac{t_{i+2}}{|x_i|} = 2p + 3 +\frac{(p+1)\left(|x_{i+1}|-|x_i|\right) + 2\hat{x}_{i-1}}{|x_i|} .
  \end{align*}
  We now assume, for the sake of contradiction, that for some $c<\sqrt8$, for all $i\geq 0$, 
  \begin{align*}
    CR_{\mathcal{A}}\left(x_i + (-1)^i\varepsilon\right) &\leq 2p + 3 + c\sqrt{p+1} \\
    \frac{(p+1)\left(|x_{i+1}|-|x_i|\right) + 2\hat{x}_{i-1}}{|x_i|} &\leq c\sqrt{p+1} \\
    (p+1)\left(|x_{i+1}|-|x_i|\right) + 2\hat{x}_{i-1} &\leq c\sqrt{p+1}|x_i| \\
    (p+1)\left(\hat{x}_{i+1}-2\hat{x}_i+\hat{x}_{i-1}\right) + 2\hat{x}_{i-1} &\leq c\sqrt{p+1}\left(\hat{x}_i-\hat{x}_{i-1}\right) \\
    \hat{x}_{i+1}-2\hat{x}_i+\hat{x}_{i-1} + \frac{2}{p+1}\hat{x}_{i-1} &\leq \frac{c}{\sqrt{p+1}}\hat{x}_i-\frac{c}{\sqrt{p+1}}\hat{x}_{i-1} \\
    \hat{x}_{i+1} \leq \left(\frac{c}{\sqrt{p+1}}+2\right)\hat{x}_i &- \left(\frac{2}{p+1}+\frac{c}{\sqrt{p+1}}+1\right)\hat{x}_{i-1}
  \end{align*}
  We invoke Lemma \myref{lem:general_becknewman} with $a=\frac{c}{\sqrt{p+1}}+2$ and $b=\frac{2}{p+1}+\frac{c}{\sqrt{p+1}}+1$. We verify that $a^2-4b<0$:
  \begin{align*}
    0 &> \left(\frac{c}{\sqrt{p+1}}+2\right)^2-4\left(\frac{2}{p+1}+\frac{c}{\sqrt{p+1}}+1\right) \\
    0 &> \frac{c^2}{p+1}+\frac{4c}{\sqrt{p+1}}+4-\frac{8}{p+1}-\frac{4c}{\sqrt{p+1}}-4 \Rightarrow 
    0 > \frac{c^2-8}{p+1} \Rightarrow
    c < \sqrt8
  \end{align*}
  By Lemma \myref{lem:general_becknewman}, a contradiction is reached as the positive sequence $\left(\hat{x}_i\right)_{i=0}^{\infty}$ cannot remain positive. We conclude that no algorithm $\mathcal{A}$ can satisfy $CR_{\mathcal{A}} < 2p+3+\sqrt{8(p+1)}$. This proves Theorem \myref{thm:odd_lowbound}.
\end{proof}

\subsection{Lower Bound for Even p}
\begin{theorem} \mylabel{thm:even_lowbound}
  For even $p$, no algorithm $\mathcal{A}$ can satisfy $CR_{\mathcal{A}}<c^*$ where $c^*$ is the unique solution to $D_p(c)=0$ in the range $\left[2p+1+\sqrt{8p},\infty\right)$.
\end{theorem}
\begin{appendixonly}
\begin{proof}
  We begin by proving Lemma \myref{lem:min_x_growth}.
  \begin{lemma} \mylabel{lem:min_x_growth}
    Consider an algorithm $\mathcal{A}$ and its corresponding sequence $\left(x_i\right)_{i=-2}^{\infty}$ obtained from $X(\mathcal{A})$.
    If $\mathcal{A}$ has competitive ratio at most $c$, then it must satisfy $|x_i| \geq \frac{2}{c}\left(\frac{c}{c-2}\right)^i$ for all $i$.
  \end{lemma}
  \begin{proof}
    After going to $x_i$ during period $i$, algorithm $\mathcal{A}$ will eventually return to the origin. Once it does so, it will have travelled for no less than $2\sum_{j=0}^{i}|x_j|$ total time (a bare minimum). At this moment, all coordinates within distance $\frac{2}{c}\sum_{j=0}^{i}|x_j|$ of the origin must have been visited so that competitive ratio $c$ is not exceeded. This means that $|x_i|$ must extend at least this far. For all $i\geq 0$,
    \begin{align}
      |x_i| &\geq \frac{2}{c}\sum_{j=0}^{i}|x_j| \mylabel{eq:min_growth_eq1}\\
      \left(1-\frac{2}{c}\right)\sum_{j=0}^{i}|x_j| &\geq \sum_{j=0}^{i-1}|x_j| \nonumber\\
      \sum_{j=0}^{i}|x_j| &\geq \left(\frac{c}{c-2}\right)\sum_{j=0}^{i-1}|x_j| \nonumber\\
      \sum_{j=0}^{i}|x_j| &\geq \left(\frac{c}{c-2}\right)^i|x_0| \geq \left(\frac{c}{c-2}\right)^i \mylabel{eq:min_growth_eq2}
    \end{align}
    Applying Equation \myeqref{eq:min_growth_eq2} to Equation \myeqref{eq:min_growth_eq1} yields that for all $i$,
    \begin{align*}
      |x_i| &\geq \frac{2}{c}\left(\frac{c}{c-2}\right)^i .
    \end{align*} 
    This proves Lemma \myref{lem:min_x_growth}.
  \end{proof}

  Consider an arbitrary algorithm $\mathcal{A}$. We obtain the sequences $\left(x_i\right)_{i=-2}^{\infty}$ and $\left(t_i\right)_{i=0}^{\infty}$ from $X(\mathcal{A})$. For each period $i$, we will select two specific coordinates and consider the performance of algorithm $\mathcal{A}$ at or near those coordinates.
  
  Our first coordinate of interest is $x_{i-2}$. For arbitrarily small values of $\varepsilon>0$, coordinate $x_{i-2}+\varepsilon$ is unexplored at the the start of period $i$. Hence, its exploration time is no less than $t_i+\varepsilon$. The associated competitive ratio is $\frac{t_i}{|x_{i-2}|+\varepsilon}$. We define $CR_1(i) = \frac{t_i}{|x_{i-2}|}$. We know that for every period $i$, a point is explored with an associated competitive ratio exceeding or getting arbitrarily close to $CR_1(i)$.

  Selecting the second key point is more involved. For each $i\geq 0$, we define $T_i$ and $y_i$ as follows:
  \begin{align*}
    T_i &= \inf\left\{t : \forall (-1)^ix\in (|x_{i-2}|,|x_i|] E_{\mathcal{A}}(x,t)\right\} ~~\text{and} \\
    y_i &= \mathcal{A}(T_i) .
  \end{align*}
  $T_i$ is the time at which algorithm $\mathcal{A}$ finishes exploring all coordinates between $x_{i-2}$ and $x_i$ (all those associated with period $i$) whereas $y_i$ is the position of those last-explored points. For any $\varepsilon>0$, there exist points within and only within distance $\varepsilon$ of $y_i$ that are unexplored at time $T_i-\varepsilon$ but are explored at time $T_i$. Such points have an associated competitive ratio no less than $\frac{T_i-\varepsilon}{|y_i|-\varepsilon}$, which may be brought arbitrarily close to $\frac{T_i}{|y_i|}$ with a sufficiently small $\varepsilon$.

  We now establish a lower bound on $\frac{T_i}{|y_i|}$. At time $T_i$, the searcher is at coordinate $y_i$. Of all the ground explored during period $i$, $|y_i|-|x_{i-2}|$ lies on the origin-facing side of $y_i$ while $|x_i|-|y_i|$ lies on the origin-opposed side of $y_i$. Exploring this ground requires $p\left(|x_i|-|x_{i-2}|\right)$ time. Additionally, ground on the origin-facing side of $y_i$ must have been retreaded at least once since the searcher crossed it an odd number of times while $p$ is even. Thus,
  \begin{equation*}
    T_i \geq t_i + p\left(|x_i|-|x_{i-2}|\right) + |y_i| - |x_{i-2}| .
  \end{equation*}
  We define $CR_2(i) = 1 + \frac{t_i-|x_{i-2}| + p\left(|x_i|-|x_{i-2}|\right)}{|y_i|}$. We know that for every period $i$, a point is explored with an associated competitive ratio exceeding or getting arbitrarily close to $CR_2(i)$. Note that after time $T_i$, the searcher must travel at least $|y_i| + |x_{i-1}|$ distance to start the next period. This gives us a lower bound on $t_{i+1}$: for all $i\geq 0$,
  \begin{align}
    t_{i+1} &\geq T_i + |y_i| + |x_{i-1}| \nonumber\\
    &\geq t_i + 2|y_i| + |x_{i-1}|-|x_{i-2}| + p\left(|x_i|-|x_{i-2}|\right) .\nonumber
  \end{align}
  This can be applied recursively. We obtain that for all $i\geq 0$,
  \begin{align}
    t_{i+1} &\geq t_0 + \sum_{j=0}^{i}\left(2|y_j|+ |x_{j-1}|-|x_{j-2}|+p\left(|x_j|-|x_{j-2}|\right)\right) \nonumber\\
    &= t_0 + 2\sum_{j=0}^{i}|y_j| + |x_{i-1}|-|x_{-2}| + p\left(|x_i|+|x_{i-1}|-|x_{-1}|-|x_{-2}|\right) \nonumber\\
    t_{i+1} &\geq |x_{i-1}| + 2\sum_{j=0}^{i}|y_j| + p\left(|x_i|+|x_{i-1}|\right) - (2p+1-t_0)  \mylabel{eq:ti_bound}
  \end{align}
  For the sake of contradiction, we now assume that algorithm $\mathcal{A}$ satisfies $CR_{\mathcal{A}}=c<c^*$ for some constant $c$ where $c^*$ is the unique value satisfying $D_p(c^*)=0$ in the range $\left[2p+1+\sqrt{8p},\infty\right)$. By Lemma \myref{lem:discrim_behavior}, $D_p(c)<0$. Recall that for each $i\geq 0$, there exist points with associated competitive ratios arbitrarily close to $CR_1(i)$ and $CR_2(i)$. As such, it is necessary that for all $i\geq 0$, $CR_1(i)\leq c$ and $CR_2(i)\leq c$. For all $i\geq 0$, the following is required to satisfy $CR_2(j)\leq c$:
  \begin{align*}
    1 + \frac{t_i-|x_{i-2}| + p\left(|x_i|-|x_{i-2}|\right)}{|y_i|} &\leq c \\
    t_i-|x_{i-2}| + p\left(|x_i|-|x_{i-2}|\right) &\leq |y_i|(c-1) \\
    2\sum_{j=0}^{i-1}|y_j| + p\left(|x_i|+|x_{i-1}|\right) - (2p+1-t_0) &\leq |y_i|(c-1) \quad\text{(Eq. \myeqref{eq:ti_bound})}\\
    (c+1)\sum_{j=0}^{i-1}|y_j| + p\left(|x_i|+|x_{i-1}|\right) - (2p+1-t_0) &\leq (c-1)\sum_{j=0}^{i}|y_j|
  \end{align*}
  Since the above equation applies for all $i\geq 0$, it can be repeatedly applied to itself.
  Let $w=\frac{c+1}{c-1}$.
  \begin{align}
    \sum_{j=0}^{i}w^{i-j}\left(p\left(|x_j|+|x_{j-1}|\right) - (2p+1-t_0)\right) &\leq (c-1)\sum_{j=0}^{i}|y_j| \nonumber\\
    p\sum_{j=0}^{i}w^{i-j}|x_j| + p\sum_{j=0}^{i}w^{i-j}|x_{j-1}| - \left(\frac{w^{i+1}-1}{w-1}\right)(2p+1-t_0) &\leq (c-1)\sum_{j=0}^{i}|y_j| \mylabel{eq:bad_neg_term}
  \end{align}
  The term $-\left(\frac{w^{i+1}-1}{w-1}\right)(2p+1-t_0)$ may take on negative values. For reasons that will eventually become clear, we need to cancel this term and have at least $\frac12(c-1)(2p+1-t_0)$ left over. We will select an index $k$ such that early terms up to $j=k$ in the above sums accomplish this.
  \begin{align*}
    p\sum_{j=0}^{i}w^{i-j}|x_j| &+ p\sum_{j=0}^{i}w^{i-j}|x_{j-1}|
    = p\sum_{j=0}^{i}w^{i-j}|x_j| + p\sum_{j=0}^{i-1}w^{i-j}|x_j| + pw^i \\
    &= p\sum_{j=k+1}^{i}w^{i-j}|x_j| + p\sum_{j=k+1}^{i-1}w^{i-j}|x_j| + p\left(w^i + 2\sum_{j=0}^{k}w^{i-j}|x_j|\right) \\
    &= p\sum_{j=k+1}^{i}w^{i-j}|x_j| + p\sum_{j=k+1}^{i-1}w^{i-j}|x_j| + pw^i\left(1 + 2\sum_{j=0}^{k}\frac{|x_j|}{w^j}\right)
  \end{align*}
  Thus, we seek
  \begin{equation*}
    pw^i\left(1 + 2\sum_{j=0}^{k}\frac{|x_j|}{w^j}\right) - \left(\frac{w^{i+1}-1}{w-1}\right)(2p+1-t_0) \geq \frac12(c-1)(2p-t_0) .
  \end{equation*}
  Note that $1-t_0\leq 0$.
  It would suffice to simply have $\frac{|x_k|}{w^k} \geq \frac{w}{w-1}+c-1 = 2c$. By Lemma \myref{lem:min_x_growth}, $|x_k|\geq \frac{2}{c}\left(\frac{c}{c-2}\right)^k$. Hence, it suffices to select
  \begin{align*}
    \frac{2}{c}\left(\frac{c}{w(c-2)}\right)^k &\geq 2c \\
    \left(1+\frac{2}{c^2-c-2}\right)^k &\geq c^2 \\
    k &\geq 2\ln\left(c\right)\left(\ln\left(1+\frac{2}{c^2-c-2}\right)\right)^{-1}
  \end{align*}
  With such a $k$ selected, Equation \myeqref{eq:bad_neg_term} may be satisfied for all $i\geq k$ by the following requirement:
  \begin{equation*}
    p\sum_{j=k+1}^{i}w^{i-j}|x_j| + p\sum_{j=k+1}^{i-1}w^{i-j}|x_j| + \frac12(c-1)(2p+1-t_0) \leq (c-1)\sum_{j=0}^{i}|y_j| .
  \end{equation*}
  Let us denote the shorthand $z_i = \sum_{j=k+1}^{i}w^{i-j}|x_j|$. Note that $x_i = z_i - wz_{i-1}$.
  \begin{equation}
    \sum_{j=0}^{i}|y_j| \geq \frac{p}{c-1}\left(z_i + z_{i-1}\right) + \frac12(2p+1-t_0) . \mylabel{eq:sumy_bound}
  \end{equation}
  With this condition in place, we switch to considering how $CR_1(j)\leq c$ is satisfied. The following transformations are valid for $i\geq k+3$:
  \begin{align*}
    \frac{t_i}{|x_{i-2}|} &\leq c \\
    t_i - c|x_{i-2}| &\leq 0 \\
    2\sum_{j=0}^{i-1}|y_j| + p\left(|x_{i-1}|+|x_{i-2}|\right) - (2p+1-t_0) - (c-1)|x_{i-2}| &\leq 0 \quad\text{(Eq. \myeqref{eq:ti_bound})}\\
    \frac{2p}{c-1}\left(z_{i-1} + z_{i-2}\right) + p\left(|x_{i-1}|+|x_{i-2}|\right) - (c-1)|x_{i-2}| &\leq 0 \quad\text{(Eq. \myeqref{eq:sumy_bound})}
  \end{align*}
  \begin{align*}
    \frac{2p}{c-1}\left(z_{i-1} + z_{i-2}\right) + p\left(z_{i-1}-wz_{i-2}+z_{i-2}-wz_{i-3}\right) - (c-1)\left(z_{i-2}-wz_{i-3}\right) &\leq 0 \\
    2p\left(z_{i-1} + z_{i-2}\right) + p(c-1)\left(z_{i-1}-wz_{i-2}+z_{i-2}-wz_{i-3}\right) - (c-1)^2\left(z_{i-2}-wz_{i-3}\right) &\leq 0 \\
    p(c+1)z_{i-1} - (c-1)^2z_{i-2} + (c+1)\left(c-1-p\right)z_{i-3} &\leq 0
  \end{align*}
  \begin{equation*}
    z_{i-1} \leq \frac{(c-1)^2}{p(c+1)}z_{i-2} - \frac{c-1-p}{p}z_{i-3}
  \end{equation*}
  We apply Lemma \myref{lem:general_becknewman} to the sequence $\left(z_{i+k+3}\right)_{i=0}^{\infty}$ where $z_{i+k+3} = \sum_{j=k+1}^{i+k+3}\left(\frac{c+1}{c-1}\right)^{i-j}|x_j|$. We take $a=\frac{(c-1)^2}{p(c+1)}$, $b=\frac{c-1-p}{p}$ and note that both are positive. We verify that $a^2-4b<0$:
  \begin{equation*}
    \left(\frac{(c-1)^2}{p(c+1)}\right)^2 - 4\left(\frac{c-1-p}{p}\right) = \frac{1}{p^2(c+1)^2}D_p(c) < 0 .
  \end{equation*}
  Our assumption that $D_p(c) < 0$ has led us to a contradiction by Lemma \myref{lem:general_becknewman}. We conclude that no algorithm $\mathcal{A}$ can satisfy $CR_{\mathcal{A}}<c^*$ where $c^*$ is the unique value satisfying $D_p(c^*)=0$ in the range $\left[2p+1+\sqrt{8p},\infty\right)$. This proves Theorem \myref{thm:even_lowbound}.
\end{proof}
\end{appendixonly}

\begin{theorem} \mylabel{thm:even_cr_bounds}
    The unique solution $c^*$ to $D_p(c)=0$ within the interval $\left[2p+1+\sqrt{8p},\infty\right)$ is bounded as follows: $2p+3+\sqrt{8(p-1)} \leq c^* \leq 2p+3+\sqrt{8p}$.
\end{theorem}
\begin{appendixonly}
\begin{proof}
    By Lemma \myref{lem:discrim_behavior}, it would suffice to show that $D_p\left(2p+3+\sqrt{8(p-1)}\right)\leq 0$ while $D_p\left(2p+3+\sqrt{8p}\right)\geq0$ within the interval in question. 
    It can be shown\footnote{We derive these results using Mathematica: \url{https://anonymous.4open.science/r/multimodal_linear_search-5DF4/}} that 
    $$D_p\left(2p+3+\sqrt{8p}\right) = 16\left(1+4\sqrt{2p}+8p+4p\sqrt{2p}+2p^2\right),$$ 
    which is clearly positive for $p\geq1$. Meanwhile, it can also be shown that 
    $$D_p\left(2p+3+\sqrt{8(p-1)}\right) = -16\left(7+4\sqrt{2(p-1)}+4p\right),$$
    which is clearly negative for $p\geq 1$. 
    This proves Theorem \myref{thm:even_cr_bounds}.
\end{proof}
\end{appendixonly}

\section{A Practical Algorithm}\mylabel{sec:approx}

The optimal algorithms described thus far use \textsc{ThoroughSearch}, which requires constantly switching directions in a way that is infinitely precise. In this section, we show that a family of modified algorithms produces competitive ratios arbitrarily close to the optimal ones without the need for infinitesimal movements or infinitely many switches in direction. The algorithms in question replace all instances of \textsc{ThoroughSearch} with another procedure that we call \textsc{CompliantThoroughSearch}. This new procedure has five parameters. Three of the parameters are identical to those of \textsc{ThoroughSearch}, describing the interval that needs to be searched and the number of modes. The fourth parameter $c$ represents the competitive ratio achieved by the algorithm that is subject to modification. The final parameter $\varepsilon>0$ is chosen by the user and represents the amount by which \textsc{CompliantThoroughSearch} is permitted to deviate from competitive ratio $c$ in its performance. Unlike \textsc{DiscreteThoroughSearch}, \textsc{CompliantThoroughSearch} partitions the searched interval into cells of unequal sizes (exponentially increasing). This is done strategically as large cells closer to the origin disproportionately affect the resulting competitive ratio. It should be noted that the way the interval is partitioned into cells does not affect the overall completion time of the procedure.

\begin{algorithm}[H]
\caption{\mylabel{alg}(Compliant Thorough Search)}
    \begin{algorithmic}[1]
        \Procedure{CompliantThoroughSearch}{$x_{init}$, $\delta$, $p$, $c$, $\varepsilon$}
            \Statex $x_{init}$: current position
            \Statex $\delta$: positive distance
            \Statex $c$: competitive ratio
            \Statex $\varepsilon$: difference in competitive ratios
            \Statex $p$: number of search modes
            \State $p' \gets p$ \textbf{if} $p$ \textbf{is odd} \textbf{else} $p + 1$
            \State $x \gets 0$ \Comment{current displacement from $x_{init}$}
            \State $s \gets \frac{\varepsilon |x_{init}|}{p' - 1}$ \Comment{cell width}
            \State \textbf{while} $x < \delta$ \textbf{do} \Indent
            \State $\Delta x \gets \min(s, \delta - x)$
            \State \textbf{call} \textsc{CellSearch}($x_{init} + x \cdot \text{sign}(x_{init})$, $\Delta x$, $p$)
            \State \textbf{if} $p$ \textbf{is even} \textbf{then} travel to $x + \Delta x$ \textbf{endif}
            \State $x \gets x + \Delta x$
            \State $s \gets s \cdot \frac{c + \varepsilon - 1}{p' - 1}$\EndIndent
            \State \textbf{endwhile}
        \EndProcedure
    \end{algorithmic}
\end{algorithm}

\begin{lemma} \mylabel{lem:replacing_thorough}
  Suppose that an algorithm $\mathcal{A}$ makes a call to \textsc{ThoroughSearch} with parameters $x_{init}$ and $\delta$ in order to explore the half-open interval composed of the set of coordinates $\left\{x : |x_{init}|<x\cdot\text{sign}\left(x_{init}\right)\leq |x_{init}|+\delta\right\}$. Suppose also that all coordinates $x$ within this interval receive $CR_{\mathcal{A}}(x)\leq c$ for some $c$. Given any $\varepsilon>0$, we can construct a modified algorithm $\mathcal{A}'$ in which the call to \textsc{ThoroughSearch} is replaced with an analogous call to \textsc{CompliantThoroughSearch} which achieves $CR_{\mathcal{A}'}(x)\leq c+\varepsilon$ for the same coordinates $x$ and which makes $n$ calls to cell search where
  \begin{equation*}
    n = \ceil*{\left(\ln\left(\frac{c+\varepsilon-1}{p'-1}\right)\right)^{-1}\cdot\ln\left(1+\frac{(c+\varepsilon-p')\delta}{\varepsilon |x_{init}|}\right)}
  \end{equation*}
  and $p'=p$ for odd $p$ while $p'=p+1$ for even $p$.
\end{lemma}
\begin{appendixonly}
\begin{proof}
    Let $n$ denote the number of cells used by \textsc{CompliantThoroughSearch}. We will label the cells with indices $0$ through $n-1$ based on the order that they are searched. Let $s_i$ denote the length of cell $i$ and let $\hat{s}_i = \sum_{j=0}^{i}s_j$. For convenience, let $b=\frac{c+\varepsilon-1}{p'-1}$. For all indices $i$ less than $n-1$, $s_i=f(i)$ where $f(i) = \frac{\varepsilon |x_{init}|}{p'-1}b^i$. We also define $F(i) = \sum_{j=0}^{i}f(j) = \frac{\varepsilon |x_{init}|}{p'-1}\cdot\frac{b^{i+1}-1}{b-1}$. The last cell with index $n-1$ has length $s_{n-1}=\delta-F(n-2)\leq f(n-1)$ so that the lengths of all the cells sum to $\delta$. As such, $n$ must be the smallest natural number satisfying $F(n-1)\geq\delta$.
    \begin{align*}
        \frac{\varepsilon |x_{init}|}{p'-1}\cdot\frac{b^n-1}{b-1} &\geq \delta \\  
        n &\geq \log_b\left(1+\frac{(p'-1)(b-1)\delta}{\varepsilon |x_{init}|}\right) \\ 
        n &= \ceil*{\left(\ln\left(\frac{c+\varepsilon-1}{p'-1}\right)\right)^{-1}\cdot\ln\left(1+\frac{(c+\varepsilon-p')\delta}{\varepsilon |x_{init}|}\right)}
    \end{align*}
    We now verify that none of the points explored produce a competitive ratio exceeding $c+\varepsilon$. An arbitrary cell $i$ contains all points of the form $|x_{init}|+\hat{s}_{i-1}<x\cdot\text{sign}\left(x_{init}\right)\leq |x_{init}|+\hat{s}_{i}$.
    Among these points, the largest $CR_{\mathcal{A}}(x)$ is produced by $|x|$ arbitrarily close to $|x_{init}|+\hat{s}_{i-1}$: those closest to the origin. For even $p$, this is obvious as these points are the last to be fully explored in the cell. For odd $p$, the searcher finishes exploring the cell in the last mode while moving away from the origin with speed $1$. Traveling with speed $1$ means that both the numerator and denominator of the competitive ratio increase by equal amounts, bringing its value closer to the minimum possible value of $1$ as the searcher progresses. Thus, the points closest to the origin still produce the largest ratio despite being searched first. Knowing this, it is sufficient to check the point $y_i = x_{init}+\hat{s}_{i-1}\text{sign}\left(x_{init}\right)$. We know that $T_{\mathcal{A}}(y_i) = T_{\mathcal{A}}(x_{init}) + p'\hat{s}_{i-1} + (p'-1)s_i$. Furthermore, we know that $T_{\mathcal{A}}(x_{init})\leq c|x_{init}|$.
    \begin{align*}
        CR_{\mathcal{A}}(y_i) &\leq \frac{c|x_{init}| + p'\hat{s}_{i-1} + (p'-1)s_i}{|x_{init}|+\hat{s}_{i-1}} \\
        &= c + \varepsilon + \frac{-\varepsilon |x_{init}| + (p'-c-\varepsilon)\hat{s}_{i-1} + (p'-1)s_i}{|x_{init}|+\hat{s}_{i-1}} \\
        &\leq c + \varepsilon + \frac{-\varepsilon |x_{init}| + \frac{p'-c-\varepsilon}{(p'-1)(b-1)}\cdot\varepsilon |x_{init}|(b^{i}-1) + \varepsilon |x_{init}|b^i}{|x_{init}|+\hat{s}_{i-1}} \\
        &= c + \varepsilon + \frac{-\varepsilon |x_{init}| - \varepsilon |x_{init}|(b^{i}-1) + \varepsilon |x_{init}|b^i}{|x_{init}|+\hat{s}_{i-1}} = c + \varepsilon .
    \end{align*}
    This completes the proof of Lemma \myref{lem:replacing_thorough}.
\end{proof}
\end{appendixonly}

\begin{theorem} \mylabel{thm:odd_compliant}
    For any odd $p\geq 3$, let $\mathcal{A}$ be an instance of Algorithm \myref{alg:odd} with the optimal parameter $a = 1+\sqrt{\frac{2}{p+1}}$ so that $CR_{\mathcal{A}} = 2p+3+\sqrt{8(p+1)}$. For any $\varepsilon>0$, a modified algorithm $\mathcal{B}$ that replaces all calls to \textsc{ThoroughSearch} with analogous calls to \textsc{CompliantThoroughSearch} with the additional parameters $c=CR_{\mathcal{A}}$ and $\varepsilon$ achieves competitive ratio $CR_{\mathcal{B}} = CR_{\mathcal{A}}+\varepsilon$ while making $O\left(\log\left(\frac{p}{\varepsilon}\right)\right)$ calls to cell search per round.
\end{theorem}
\begin{appendixonly}
\begin{proof}
    Algorithm \myref{alg:odd} makes one call to \textsc{ThoroughSearch} every round $i$ with $x_{init} = (-1)^ia^{i-1}$ and $\delta = a^{i-1}(a^2-1)$. 
    By lemma \myref{lem:replacing_thorough}, replacing this call with \textsc{Compliant\allowbreak Thorough\allowbreak Search} and adding the parameters $c=CR_{\mathcal{A}}$ and $\varepsilon$ ensures that the same area is searched in the same amount of time with the new competitive ratio $CR_{\mathcal{A}}+\varepsilon$ and with $n$ calls to cell search where
    \begin{align*}
        n &= \ceil*{\left(\ln\left(\frac{c+\varepsilon-1}{p-1}\right)\right)^{-1}\cdot\ln\left(1+\frac{(c+\varepsilon-p)\delta}{\varepsilon |x_{init}|}\right)} \\
        &= \ceil*{
            \begin{aligned}
                &\left(\ln\left(\frac{2p+2+\sqrt{8(p+1)}+\varepsilon}{p-1}\right)\right)^{-1} \\
                &\quad \cdot \ln\left(1+\frac{(p+3+\sqrt{8(p+1)}+\varepsilon)(a^2-1)}{\varepsilon}\right)
            \end{aligned}
        } \\
        &\leq \ceil*{\left(\ln(2)\right)^{-1}\cdot\ln\left(1+\frac{(p+p+3(p+1)+\varepsilon)\left(\left(1+\sqrt{\frac{2}{p+1}}\right)^2-1\right)}{\varepsilon}\right)} \\
        &\leq \ceil*{\left(\ln(2)\right)^{-1}\cdot\ln\left(1+\frac{(p+p+3p+p+\varepsilon)\left((2)^2-1\right)}{\varepsilon}\right)} \\
        &= \ceil*{\left(\ln(2)\right)^{-1}\cdot\ln\left(4+\frac{18p}{\varepsilon}\right)}
        = O\left(\log\left(\frac{p}{\varepsilon}\right)\right) .
    \end{align*}
    This concludes the proof of Theorem \myref{thm:odd_compliant}.
\end{proof}
\end{appendixonly}

\begin{theorem} \mylabel{thm:even_compliant}
    For any even $p$, let $\mathcal{A}$ be an instance of Algorithm \myref{alg:even} with parameters $a$ and $r$ such that it yields the optimal competitive ratio $CR_{\mathcal{A}}$ according to Theorem \myref{thm:even_upbound}. For any $\varepsilon>0$, a modified algorithm $\mathcal{B}$ that replaces all calls to \textsc{ThoroughSearch} with analogous calls to \textsc{CompliantThoroughSearch} with the additional parameters $c=CR_{\mathcal{A}}$ and $\varepsilon$ achieves competitive ratio $CR_{\mathcal{B}} = CR_{\mathcal{A}}+\varepsilon$ while making $O\left(\log\left(\frac{p}{\varepsilon}\right)\right)$ calls to cell search per round.
\end{theorem}
\begin{appendixonly}
\begin{proof}
    Algorithm \myref{alg:even} makes one call to \textsc{ThoroughSearch} every round $i$ with $x_{init} = (-1)^ia^{i-1}$ and $\delta = a^{i-1}r(a^2-1)$. By lemma \myref{lem:replacing_thorough}, replacing this call with \textsc{Compliant\allowbreak Thorough\allowbreak Search} and adding the parameters $c=CR_{\mathcal{A}}$ and $\varepsilon$ ensures that the same area is searched in the same amount of time with the new competitive ratio $CR_{\mathcal{A}}+\varepsilon$ and with $n$ calls to cell search where
    \begin{align*}
        n &= \ceil*{\left(\ln\left(\frac{c+\varepsilon-1}{p}\right)\right)^{-1}\cdot\ln\left(1+\frac{(c+\varepsilon-p-1)\delta}{\varepsilon |x_{init}|}\right)}\\
        &= \ceil*{\left(\ln\left(\frac{c+\varepsilon-1}{p}\right)\right)^{-1}\cdot\ln\left(1+\frac{r(c+\varepsilon-p-1)(a^2-1)}{\varepsilon}\right)}
    \end{align*}
    Theorem \myref{thm:even_upbound} specifies that $a = \frac{(c-1)^2}{2p(c+1)}$. From Theorem \myref{thm:even_cr_bounds}, we know that $2p+3+\sqrt{8(p-1)}\leq c\leq 2p+3+\sqrt{8p}$. This can be used to derive an upper bound for $a$.
    \begin{equation*}
        a \leq \frac{\left(2p+2+\sqrt{8p}\right)^2}{2p\left(2p+4+\sqrt{8(p-1)}\right)}
        \leq \frac{\left(2p+p+3p\right)^2}{2p\left(2p\right)} = 9
    \end{equation*}
    These bounds, as well as the fact that $r\leq1$, yield the following:
    \begin{align*}
        n &\leq \ceil*{\left(\ln\left(\frac{2p+2+\sqrt{8(p-1)}+\varepsilon}{p}\right)\right)^{-1}\cdot\ln\left(1+\frac{(p+2+\sqrt{8p}+\varepsilon)(80)}{\varepsilon}\right)} \\
        &\leq \ceil*{\left(\ln(2)\right)^{-1}\cdot\ln\left(81+\frac{400p}{\varepsilon}\right)}
        = O\left(\log\left(\frac{p}{\varepsilon}\right)\right) .
    \end{align*}
    This concludes the proof of Theorem \myref{thm:even_compliant}.
\end{proof}
\end{appendixonly}

\section{Conclusion}
\mylabel{sec:conclusion}

In this paper, we introduced a new type of linear search, multimodal linear search, with $p$ search modes, a generalization of the well-known linear search with $p=1$.
We proposed algorithms (one for the case of odd $p$ and one for the case of even $p$) which we show achieve the optimal competitive ratio for multimodal search for an unknown target by a single searcher.
As anticipated, our results generalize the well-known competitive ratio of $9$ for the standard linear search problem (i.e., $p=1$). 

We also proposed modified algorithms (with competitive ratio $c + \varepsilon$ where $c$ is the optimal competitive ratio) that involve the searcher moving finite distances and switching directions a finite number of times over any finite interval.
This suggests that it may be interesting to study the multimodal search problem for the case where there is a cost associated with switching directions and/or search modes.
In this variant, Algorithms~\ref{alg:odd} and~\ref{alg:even} would not only be impractical, but have unbounded competitive ratio.
Turn cost has been studied before~\cite{gal_search_games}[Section 8.4, p 132],~\cite{turn_cost,turn_cost_gal,lopezsearching} but, to the best of our knowledge, has never been studied in a multimodal context.

Many other interesting research questions remain open, including randomized algorithms, group search (with multiple cooperating searchers), fault tolerance of the searchers, as well as search on domains more general than the real line (e.g., a star of infinite rays).

\bibliographystyle{abbrv}
\bibliography{main}





\end{document}